\documentclass[copyright,creativecommons]{eptcs}
\providecommand{\event}{ICE 2010} 

\usepackage{verbatim}
\usepackage{tikz}
\usetikzlibrary{shapes}
\usetikzlibrary{automata}
\usetikzlibrary{patterns}
\usepackage{tikzordonnancement}
\usepackage{mdwlist}
\usepackage{subfigure}
\usepackage{latexsym}
\usepackage{amsthm}
\usepackage{amssymb}
\usepackage[utf8]{inputenc}
\usepackage[T1]{fontenc}

\title{An Introduction to Time-Constrained Automata}

\author{Matthieu Lemerre, Vincent David and Christophe Aussaguès 
\institute{CEA, LIST, Embedded Real-Time System Lab\\
91191 Gif-sur-Yvette Cedex, France}
\email{Firstname.Lastname@cea.fr}
\and Guy Vidal-Naquet
\institute{SUPELEC\\
91192 Gif-sur-Yvette Cedex, France}
\email{Guy.Vidal-Naquet@supelec.fr}}
\date{}

\def\titlerunning{An Introduction to Time-Constrained Automata}
\def\authorrunning{M. Lemerre, V. David, C. Aussaguès \&  G. Vidal-Naquet}
\begin{document}

\newenvironment{todo}{\hrule{\bf ptodo}}{\hrule}

\newcommand{\mytodo}[1]{\textcolor{red}{#1}}
\newcommand{\xxx}[1]{\textcolor{red}{#1}}
\newcommand{\mayreduce}[1]{}
\newenvironment{unclean}{\hrule{\bf unclean}\color{red}}{\hrule}
\newenvironment{obsolete}{\comment}{\endcomment}

\newenvironment{reduce}{\comment}{\endcomment}
\newenvironment{noreduce}{}{}

\newcommand{\paragraphR}[1]{\paragraph{#1}}

\theoremstyle{definition} \newtheorem{definition}{Definition}
\theoremstyle{definition} \newtheorem{example}[definition]{Example}
\theoremstyle{plain} \newtheorem{theorem}[definition]{Theorem}
\theoremstyle{plain} \newtheorem{lemma}[definition]{Lemma}
\theoremstyle{plain} \newtheorem{corollary}[definition]{Corollary}

\colorlet{a}{red!50}
\colorlet{b}{green!60!blue!60}
\colorlet{c}{blue!40}
\colorlet{d}{yellow!80}

\newcommand{\psic}{$\Psi{}$C}


\tikzstyle{invisible}=[circle, inner sep=0pt, minimum size=0pt]
\tikzstyle{block}=[->, shorten >=1pt,>=stealth, semithick]
\tikzstyle{noconstraint}=[draw,circle, inner sep=4pt, semithick, fill=black!20]
\tikzstyle{after}=[draw, inner sep=3pt, shape=isosceles triangle, thick, draw=blue!50, fill=blue!20]
\tikzstyle{before}=[draw, inner sep=3pt, shape=isosceles triangle,shape border rotate=180, thick, draw=blue!50, fill=blue!20]
\tikzstyle{synchronisationpoint}=[draw,diamond, inner sep = 4pt, thick, draw=blue!50, fill=blue!20]

\tikzstyle{anyconstraint}=[draw,rectangle]


\maketitle{}

\begin{abstract}
  We present time-constrained automata (TCA), a model for hard
  real-time computation in which agents behaviors are modeled by
  automata and constrained by time intervals.

  TCA actions can have multiple start time and deadlines, can be
  aperiodic, and are selected dynamically following a graph, the
  time-constrained automaton. This allows expressing much more precise
  time constraints than classical periodic or sporadic model, while
  preserving the ease of scheduling and analysis.

  We provide some properties of this model as well as their scheduling
  semantics. We show that TCA can be automatically derived from
  source-code, and optimally scheduled on single processors using a
  variant of EDF. We explain how time constraints can be used to
  guarantee communication determinism by construction, and to study
  when possible agent interactions happen.
\end{abstract}

\section{Introduction}

Most concrete implementations of real-time systems use only two
different kinds of tasks: \emph{periodic tasks} and \emph{sporadic
  tasks}. Periodic tasks must execute one job per period of time, and
are meant for regular processing. Sporadic tasks are meant for
processing of events of limited occurrence, with jobs having a
deadline relative to the arrival of the event. We believe that these
kinds of tasks are not expressive enough in many situations.
Restricting real-time systems to use only them puts heavy constraints
on the design of real-time applications, which make them harder to
design, implement and analyze.

For instance, the timing behavior of complex tasks can be specified
using timed automaton\cite{alur94theory}, whose behavior is not
cyclic. These kind of tasks fit neither the periodic nor the sporadic
task model, making the translation to these tasks inefficient. Common
example of complex timing behaviors are degraded mode, multi-phase
applications (e.g. take-off/flight/landing phases of air
travel\ldots). An issue that motivates the need for more accurate task
models is \emph{jitter}, which is the variation between successive
executions of a periodic task. Current real-time methodologies consist
of analyzing jitter once the design is done and the execution times of
the tasks are known. Thus, the whole design has to be modified if the
jitter of a task is too high. It is also impossible to take into
account the fact that different tasks have different degrees of
sensitivity to jitter. On the contrary, the time-constrained task
model we present allows to express the maximum jitter directly in the
model, making it a constraint that the scheduling algorithm has to
enforce. Moreover the maximum jitter can appear in the specification,
design and code of each task, and bounding of jitter is thus
guaranteed by construction.

We claim that using a more expressive task model allows an easier
development (less transformations need to be done from specification
to the implementation) and better verification, thus increasing the
safety of the system. This is why OASIS \cite{chabrol05deterministic},
a toolchain implementing a subset of the time-constrained task model,
is used in hard real-time safety-critical environments, such as the
nuclear industry \cite{david04oasis}. This toolchain comprises in
particular a specific compiler, a microkernel and operating system
services, whose behaviors are functionally described in this article.
We also claim that using this model allows to reduce hardware costs,
because accurate modeling of task needs and exact feasibility analysis
allows to achieve a very high utilization, on single and multiple
processors.

The purpose of this paper is to present formal semantics of the
time-constrained automata model, and proof of some important results
(optimality of EDF, communication determinism) that come from using
this model. The paper is structured as follows:
Section~\ref{sec:related-work} present related work, and
Section~\ref{sec:time-constrained-automata} the time-constrained task
model. Section~\ref{sec:applications} provides example uses and shows
how they can be derived automatically from source code expressed in a
suitable language. Section~\ref{sec:scheduling-time-constrained-tasks}
shows how time-constrained tasks can be scheduled, and gives an
optimal scheduling algorithm on single processors based on EDF.

\section{Related work}
\label{sec:related-work}
\paragraph{Timed automata} Time-constrained automata is a model to be
used for scheduling rather than for accurate model-checking. Notable
restrictions from model-checking models such as timed automata
\cite{alur94theory} is that we abstract the control flow logic by
considering that the choice between several transitions is
nondeterministic. Moreover, it is the choice of a transition that acts
on a time constraint, rather than having time constraints acts on
the control flow logic.

However, time-constrained automata can be modeled using timed
automata, using only two clocks. It is thus a simpler model, i.e. less
expressive, but easier to analyze, than the general timed automata. 

\paragraph{Task models}

The main characteristics of our model is that it allows to express
infinite computations and allows dynamic modification of time
constraints. Usual task models perform dynamic release of fixed jobs
(i.e. the start time, deadline of the execution time does not evolve)
\cite{baruah04schedulingreal-time}. By contrast, a time-constrained
automaton can be viewed as one job that changes dynamically.

Moreover, timing behavior can change depending on choices not
expressed in the automaton: so the task describes a set of possible
sets of timing requirements, rather than a fixed set of timing
requirements.

There have been other attempts to provide more accurate real-time
models with multiple deadlines and start times \cite{baruah98general},
but they are event-triggered and impose cyclic behavior.

Some work on finite computations task models takes into account
dependencies between jobs \cite{muntz70preemptive} or start times and
deadlines \cite{horn74somesimple}, but none had dynamic changes of
timing behavior (i.e. jobs always have one fixed start time and one
fixed deadline).

Finally, some results exist in the literature about scheduling
independent of the task model, but they often assume that jobs do not
change dynamically \cite{baruah04schedulingreal-time}. However, these
results can be adapted to fit our model, as it is done here to the
proof that EDF is optimal \cite{dertouzos74control}.

\paragraph{Language interface to time constraints}

The \psic{} language allows to express tasks in this model
(Section~\ref{sec:writing-time-constrained-programs}) by indicating
time constraints using a language interface, rather than using an
API (like POSIX). This allows the application designer to focus on his
needs rather than on scheduling.

This is similar to synchronous programming languages like Esterel
\cite{berry00foundations} or Lustre \cite{halbwachs91synchronous}, or
time-triggered programming language such as
Giotto~\cite{henzinger01giotto}. The main diffence of our work with
these languages is the underlying task model, which is not restricted
to periodic or sporadic tasks, but rather add timing constraints to
any automata.

\paragraph{Time-triggered architecture and models} The time-triggered 
architecture~\cite{kopetz98time-triggered} focus on separation of a
hard real-time system between interfaces and components. The interface
consist of asynchronous communication at a priori known sampling
dates. The components in this architecture is often periodic or
sporadic.

\section{Time-constrained tasks} 
\label{sec:time-constrained-automata}

The time-constrained task model allows accurate description of the
time constraints of single-threaded time-constrained computations
using graphs called \emph{time-constrained automata}. Concurrent
time-con\-strained computations are represented by multiple
automata.

We present time-constrained automata in three steps: introduction to the
model given for chains and trees before explaining the time-constrained
automata.

\subsection{Chains}

\paragraph{Blocks, arcs, nodes and their relationships}

A \emph{block} is a sequence of instructions taken as a whole (the
term ``block'' comes from the control flow graph terminology). Blocks
are represented by \emph{arcs} and are separated by \emph{nodes}. The
node from which the arc starts \emph{immediately precedes} the arc,
and the node to which it leads \emph{immediately succeeds} the arc.

From the relations \emph{immediately precedes} and \emph{immediately
  succeeds}, we can derive by transitive closure the relations
\emph{precedes} and \emph{succeeds}. It can be said that an arc
succeeds an arc, an arc succeeds a node, a node succeeds an arc, or a
node succeeds a node (resp. precedes).

A \emph{chain} is a sequence of blocks, executing one after the other.
When blocks $a$ and $b$ are consecutive, instructions of $a$ have to
be executed before those of $b$. 

Chains have a \emph{first node}, at which they start, but can be
\emph{infinite}.

\paragraph{Temporal constraints}

As a chain is a sequence of indivisible blocks, time constraints can
apply only to blocks. Only two kind of constraints are possible:

\begin{itemize*}
\item A block can be constrained to start only \emph{after} a certain
  \emph{date} (i.e. specific time instant), or
\item it can be constrained to finish \emph{before} one.
\end{itemize*}

As TCA target hard real-time systems, time constraints are strict,
i.e. failure to meet them is an incorrect behavior. We choose to make
nodes bear the constraints:
\begin{itemize*}
\item ``After constraints'' of a block are borne by the immediate
  predecessor node, denoted by~$\rhd$;
\item ``Before constraints'' of a block are borne by the immediate
  successor node, denoted by~$\lhd$;
\item When a block bears both a before and an after constraint
  \emph{at the same date}, then it becomes a \emph{synchronization
    point}, represented by $\diamondsuit$.
\item Except for synchronization points, nodes cannot bear more than
  one constraint\footnote{Allowing more than one constraint per node
    would not extend the expressive power, because of constraints
    redundancy seen below.}.
\item Nodes with no constraint are represented by {\small $\bigcirc$}.
\end{itemize*}

We label the constraint nodes by the absolute date that they
represent. Figure~\ref{fig:constrained-chaine-possible-schedule} gives
an example.

\begin{figure}[htbp]
  \centering
  
  \begin{tikzpicture}[xscale=1.8]

    \begin{scope}[xscale=0.5]
    \node[after,label=below:1] (1) at (1,0) {};
    \node[after,label=below:2] (2) at (2,0) {};
    \node[before,label=below:5] (5) at (5,0) {};
    \node[synchronisationpoint,label=below:7] (7) at (7,0) {};
    \node[before,label=below:10] (10) at (10,0) {};

    \draw[block] (1) -- node[above] {a} (2);
    \draw[block] (2) -- node[above] {b} (5);
    \draw[block] (5) -- node[above] {c} (7);
    \draw[block] (7) -- node[above] {d} (10);

    \foreach \x in {1,2,5,7,10} 
    \draw[dashed, help lines] (\x, 1) -- (\x, -2);

    \begin{scope}[shift={(0,-1.5)},yscale=0.5,fill=gray!80]
      \draw[fill=black!10] (1, 0) rectangle (2.3, 1);
      \draw[fill=black!25] (2.3, 0) rectangle (3.5, 1);
      \draw[fill=black!25] (4.1, 0) rectangle (4.8, 1);
      \draw[fill=black!34] (4.8, 0) rectangle (6.4, 1);
      \draw[fill=black!45] (7.3, 0) rectangle (9.1, 1);

      \node at (1.65, 0.5) {a};
      \node at (2.9, 0.5) {b};
      \node at (4.45, 0.5) {b};
      \node at (5.6, 0.5) {c};
      \node at (8.2, 0.5) {d};
    \end{scope}
  \end{scope}
  
  \end{tikzpicture}

  \caption{Constrained chain: $a$ must start after date 1, $b$ must
    execute between 2 and 5, $c$ must end before 7, and $d$ must
    execute between 7 and 10. Beneath is a possible corresponding
    preemptive schedule. }
  \label{fig:constrained-chaine-possible-schedule}
\end{figure}
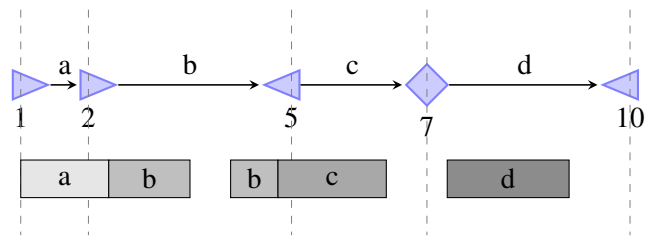

\paragraph{Extension of constraints to other arcs }

An \emph{after} node implicitly constrains all the succeeding blocks
to start after its date, and a \emph{before} node constrains all the
preceding blocks to end before its date. This is formally stated by
the following lemma:

\begin{lemma}[Implicit extension of constraints to other arcs]
  \label{th:implicit-extension-constraints-all-arcs}
  If a block $b$ succeeds an after node $A$ of date $d$, then $b$ must
  start after $d$.

  If a block $b$ precedes a before node $B$ of date $d'$, then $b$ must
  end before $d'$.
\end{lemma}

\begin{proof} The lemma is proved by recurrence:
  
  Either $b$ immediately succeeds $A$, then by definition $b$ must
  start after $d$.

  Else assume the lemma is true for the block $c$ at distance $n$ from
  $A$, i.e. $c$ must start after date $d$; then the node at distance
  $n+1$ from $A$ (if any) immediately succeeds block $c$, so executes
  after it, and must start after $d$.

  The proof is similar for the \emph{before} node.
\end{proof}

From the graphical representation, one can easily derive the implicit
constraints on a block. For instance in
Figure~\ref{fig:constrained-chaine-possible-schedule}, $a$ must end
before 5, because the ``5'' before constraint succeeds $a$.

Thus time-constrained automata are a model based on \emph{possible
  intervals of computations}, i.e. each block has a particular
interval during which it can execute. By contrast, the synchronous
model is based on single \emph{points} of computation.

\paragraph{Impossible and redundant constraints}

There are also some relationships between the dates of the different
constraints: some cannot be satisfied, others can be simplified
because they are implied by another one.

Figure~\ref{fig:nodes-succeeding-later-node} shows the different cases
where a constraint is followed by a constraint with an earlier date.

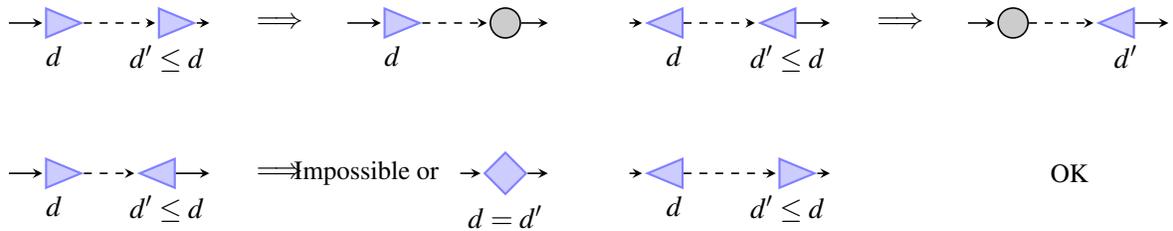
\begin{figure}[htbp]
  \centering
  
  \begin{tikzpicture}[xscale=1.5]

    \begin{scope}
    \node[after,label=below:$d$] (a) at (0,0) {};      
    \node[after,label=below:$d'\le d$] (b) at (1,0) {};      

    \node at (2,0) {$\Longrightarrow$};

    \node[after,label=below:$d$] (a') at (3,0) {};      
    \node[noconstraint] (b') at (4,0) {};      

    \draw[block] (a) +(-0.4, 0) -- (a);
    \draw[block,dashed] (a) -- (b);
    \draw[block] (b) -- + (0.4, 0);

    \draw[block] (a') +(-0.4, 0) -- (a');
    \draw[block,dashed] (a') -- (b');
    \draw[block] (b') -- + (0.4, 0);
    \end{scope}

    \begin{scope}[yshift=-2cm]
    \node[after,label=below:$d$] (a) at (0,0) {};      
    \node[before,label=below:$d'\le d$] (b) at (1,0) {};

    \node at (2,0) {$\Longrightarrow$};
    \node[left] at (3.5,0) {{\small Impossible or }};
    \node[synchronisationpoint,label={below:$d=d'$}] (a') at (4,0) {};

    \draw[block] (a) +(-0.4, 0) -- (a);
    \draw[block,dashed] (a) -- (b);
    \draw[block] (b) -- + (0.4, 0);

    \draw[block] (a') +(-0.4, 0) -- (a');
    \draw[block] (a') -- + (0.4, 0);

    \end{scope}

    \begin{scope}[yshift=0cm, xshift= 5.5cm]
    \node[before,label=below:$d$] (a) at (0,0) {};      
    \node[before,label=below:$d'\le d$] (b) at (1,0) {};      

    \node at (2,0) {$\Longrightarrow$};

    \node[noconstraint] (a') at (3,0) {};      
    \node[before,label=below:$d'$] (b') at (4,0) {};      

    \draw[block] (a) +(-0.4, 0) -- (a);
    \draw[block,dashed] (a) -- (b);
    \draw[block] (b) -- + (0.4, 0);

    \draw[block] (a') +(-0.4, 0) -- (a');
    \draw[block,dashed] (a') -- (b');
    \draw[block] (b') -- + (0.4, 0);
    \end{scope}

    \begin{scope}[yshift=-2cm, xshift =5.5cm]
    \node[before,label=below:$d$] (a) at (0,0) {};      
    \node[after,label=below:$d'\le d$] (b) at (1,0) {};      

    \node at (3.5,0) {{\small OK}}; 

    \draw[block] (a) +(-0.4, 0) -- (a);
    \draw[block,dashed] (a) -- (b);
    \draw[block] (b) -- + (0.4, 0);
    \end{scope}

  \end{tikzpicture}

  \caption{Possible simplifications when a node precedes a node with
    earlier date.}

  \label{fig:nodes-succeeding-later-node}
\end{figure}

\begin{theorem} \label{th:useless-undoable-constraints}
  When a node $N$ of date $d$ precedes a node $N'$ of date $d' \le d$,
  then
  \begin{itemize*}
  \item If $N$ and $N'$ are after nodes, then $N'$ can be removed.
  \item If $N$ and $N'$ are before nodes, then $N$ can be removed.
  \item If $N$ is an after node and $N'$ a before node, then either
    both can be merged into a synchronization point or the constraints
    are impossible to fulfill.
   \end{itemize*}
\end{theorem}

\begin{proof}

  \begin{itemize}
  \item If $N$ and $N'$ are both after constraints: then by
    Lemma~\ref{th:implicit-extension-constraints-all-arcs}, the block
    immediately succeeding $N'$ is already implicitly constrained to
    start after $d \ge d'$. So the constraint is redundant.

  \item Similarly, if $N$ and $N'$ are both before constraints,
    Lemma~\ref{th:implicit-extension-constraints-all-arcs} implies
    that the $N$ constraint is already implied by $N'$, and is
    redundant.

  \item If $N$ is an after constraint and $N'$ a before constraint:
    \begin{itemize}
    \item If $d' < d$, then any block between $N$ and $N'$ is
      constrained to start after it ends, which is a condition
      impossible to fulfill.
    \item Else $d = d'$: either a block between $N$ and $N'$ is not
      empty (i.e. has at least one instruction), in which case we must
      execute some instructions in 0 time, which is impossible. Else
      all blocks between $N$ and $N'$ are empty. Then we can replace
      $N$, $N'$ and the blocks between by a single synchronization
      node of date $d$: by
      Lemma~\ref{th:implicit-extension-constraints-all-arcs}, the
      constraints of all blocks after $N'$ and before $N$ are
      preserved.
    \end{itemize}
  \end{itemize}
\end{proof}

Thus, there is only one case where it is useful that a node precedes a
node of earlier date, which is the fourth case of
Figure~\ref{fig:nodes-succeeding-later-node}.

\paragraph{Relative labeling of constraints}

The following is an immediate corollary of
Theorem~\ref{th:useless-undoable-constraints}:

\begin{corollary}
  A chain can be simplified so that the dates of all constraint nodes
  following an after node of date $d$ are greater than $d$.
\end{corollary}

Thus the following labelling convention can be adopted without
modifying the semantics of our model: \emph{all node dates can be
  labeled using the relative date from the previous after node}
(including synchronization points). This convention will allow
expression of loops in automata. We denote the use of relative
labeling by putting underscores below dates. The first node of a chain
is labeled relatively from 0.

Figure~\ref{fig:relative-labelling} represents the re-labeling of
Figure~\ref{fig:constrained-chaine-possible-schedule} with relative
dates.

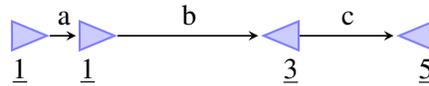
\begin{figure}[htbp]
  \centering
  
  \begin{tikzpicture}[xscale=1.8]

    \begin{scope}[xscale=0.5]
    \node[after,label=below:{\underline 1}] (1) at (1,0) {};
    \node[after,label=below:{\underline 1}] (2) at (2,0) {};
    \node[before,label=below:{\underline 3}] (5) at (5,0) {};
    \node[before,label=below:{\underline 5}] (7) at (7,0) {};

    \draw[block] (1) -- node[above] {a} (2);
    \draw[block] (2) -- node[above] {b} (5);
    \draw[block] (5) -- node[above] {c} (7);
  \end{scope}
  
  \end{tikzpicture}

  \caption{Chain of
    Figure~\ref{fig:constrained-chaine-possible-schedule} with
    relative labeling}
  \label{fig:relative-labelling}
\end{figure}

Chains can be used to model the history of a job release by a task.
However they have some obvious limitations: infinite computations can
only be modeled by infinite chains, which cannot be written in a
specification; and chains cannot model conditional execution of
blocks.

\subsection{Time-constrained trees}

We extend the previous concept of chains to \emph{trees}, which
requires to handle ``choices''.

We now allow several blocks to start from a node (such a node is
called a \emph{choice node}). This expresses the fact that different
execution paths may be taken. The choice of which path to take is made
when finishing executing the immediately preceding block.
Figure~\ref{fig:addition-of-choice} gives an example.

\begin{figure}[htbp]
  \centering

  \begin{tikzpicture}[xscale=1.8]

    \node[after,label=below:{\underline 1}] (1) at (0,0) {};
    \node[noconstraint] (2) at (1,0) {};
    \node[before,label=below:{\underline 3}] (5) at (2,0) {};
    \node[before,label=below:{\underline 2}] (7) at (2,-1) {};

    \draw[block,dashed] (1) +(-0.4,0) -- (1);
    \draw[block] (1) -- node[above] {a} (2);
    \draw[block] (2) -- node[above] {b} (5);
    \draw[block] (2) |- node[above,near end] {c} (7);

    \draw[block,dashed] (5) --  +(0.4,0);
    \draw[block,dashed] (7) --  +(0.4,0);
  
  \end{tikzpicture}

  \caption{Depending on the execution of $a$, either $b$ or $c$ will
    be executed. The path $a\to b$ has 2 units of time to complete,
    but the path $a\to c$ only has 1.}
  \label{fig:addition-of-choice}
\end{figure}
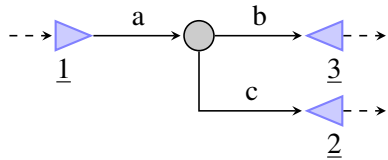

All paths in the graph constitutes chains, and the ``precedes'' and
``succeeds'' relationships are still defined. Thus the previous
theorems and lemmas that applied to chains still hold.

\subsection{Automata}

To represent infinite computations with finite objects, we use
automata. Automata differ from trees in two aspects: first they allow
several arcs to finish on the same node; second they allow cycles in
the graph. These differences change the \emph{precedes} relation on
blocks, and affects negatively the semantics, which depends on this
relation. We address these problems by defining the semantics of
automaton as such: the semantics of a time-constrained automaton is
equivalent to that of its \emph{unfolded tree}, which is the infinite
tree representing all possible traversals of the automaton. Thus, the
\emph{precedes} relation, and the semantics, are preserved.

Moreover, time in automata is expressed using relative labelling, and
the unfolding operation also performs the conversion to absolute
labelling.

Figure~\ref{fig:unfolding-automaton} shows example of such an
unfolding.

\begin{figure}[htbp]
  \centering

  {
 \begin{tikzpicture}[xscale=1.6, scale=1.1,yscale=0.8]

   \begin{scope}[shift={(-3,-1)}]

   \node[noconstraint] (A) at (0,0) {};
   \node[synchronisationpoint,label=below:{\underline 2}] (B) at (-1,0) {};
   \node[synchronisationpoint,label=below:{\underline 1}] (C) at (1,0) {};

   \draw[block] (A) +(0,0.4) -- (A);
     \draw[block] (A) edge[bend left=45]  node[below] {a} (B);
     \draw[block] (B) edge[bend left=45] node[above] {b} (A);
     \draw[block] (A) edge[bend right=45]  node[below] {c} (C);
     \draw[block] (C) edge[bend right=45] node[above] {d} (A);

     \node at (2,0) {{\LARGE \bf =}};
     
   \end{scope}

     \begin{scope}

       \node[noconstraint] (A1) at (0,0) {};
       \node[synchronisationpoint,label=right:2] (B1) at (0,-1) {};
       \node[synchronisationpoint,label=below:1] (C1) at (1,0) {};
       \node[noconstraint] (A2) at (2,0) {};
       \node[synchronisationpoint,label=right:3] (B2) at (2,-1) {};
       \node[synchronisationpoint,label=below:2] (C2) at (3,0) {};
       \node[noconstraint] (A3) at (0,-2) {};
       \node[synchronisationpoint,label=right:4] (B3) at (0,-3) {};
       \node[synchronisationpoint,label=below:3] (C3) at (1,-2) {};

       \node[invisible] (A4) at (4, 0) {};
       \node[invisible] (A5) at (2, -1.8) {};
       \node[invisible] (A6) at (1.5, -2) {};
       \node[invisible] (A7) at (0, -4) {};

       \draw[block] (A1) +(0,0.4) -- (A1);
       \draw[block] (A1) edge node[above] {c} (C1);
       \draw[block] (C1) edge node[above] {d} (A2);
       \draw[block] (A1) edge node[left] {a} (B1);
       \draw[block] (B1) edge node[left] {b} (A3);

       \draw[block] (A2) edge node[above] {c} (C2);
       \draw[block,dashed] (C2) edge node[above] {d} (A4);
       \draw[block] (A2) edge node[left] {a} (B2);
       \draw[block,dashed] (B2) edge node[left] {b} (A5);

       \draw[block] (A3) edge node[above] {c} (C3);
       \draw[block,dashed] (C3) edge node[above] {d} (A6);
       \draw[block] (A3) edge node[left] {a} (B3);
       \draw[block,dashed] (B3) edge node[left] {b} (A7);

     \end{scope}

   \end{tikzpicture}}

  \caption{Unfolding of an automaton}
  \label{fig:unfolding-automaton}
\end{figure}
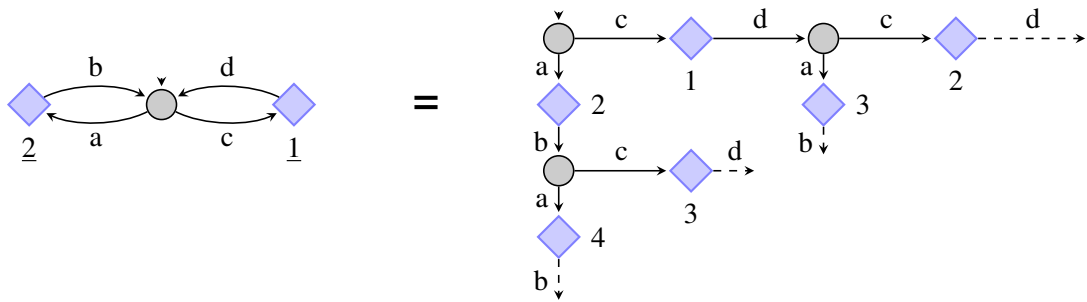

This semantics of automata using ``unfolding'' allows to preserve the
``precede'' and ``succeed'' relationships as defined on trees, which
are the basis for the semantics of time-constrained computations.

Note that all time-constrained computations cannot be represented by a
finite automaton. For instance, an algorithm executing a block in
loop, with the $k^{\mathrm{th}}$ iteration constrained to execute
between instants $2^{k}$ and $2^{k+1}$ (because relative constraints
can only be added). This case is of little practical interest though;
in fact, we believe that most interesting cases are representable with
automata (with only one extension not shown here for the sake of
simplicity, which is the synchronization with a different clock).

\section{Applications}
\label{sec:applications}

In this section we present some basic examples of time-constrained
automata application, as well as an implementation based on the
\psic{} programming language.

\subsection{Example uses}

\begin{figure}[htbp] 
\centering  \hbox{\hspace{0cm}
  \subfigure[Periodic task of period 2 with relative deadline equal to the period.]{
    \begin{tikzpicture}    
      \draw[line width=0pt,white] (-1.5,0) -- (1.5,0);

    \node[synchronisationpoint,label=below:{\underline 2}] (P1) at (0,0) {};
    \draw[block] (P1) edge[loop] (P1);
     \end{tikzpicture}}
   \hspace{5mm}

   \subfigure[A periodic task (period 5). $b$ is constrained with
   fine-grained jitter specification (maximum jitter 1).]{
    \begin{tikzpicture}    
        \draw[line width=0pt,white] (-1,0) -- (3,0);
      \begin{scope}[xscale=1.2,xshift=-0.3cm]
        \node[synchronisationpoint,label=below:{\underline 2}] (N1) at (0,0) {};
        \node[after,label=below:{\underline 3}] (N2) at (1,0) {};
        \node[before,label=below:{\underline 1}] (N3) at (2,0) {};

        \node[invisible,label=below:$c$] (a) at (2.2,-1.2) {};

        \draw[block] 
        (N1) edge node[above] {$a$} (N2)
        (N2) edge node[above] {$b$} (N3);
        
        \draw (N3.south) edge[bend left] (a);
        \draw[block,->] (a) edge[bend left] (N1.south);

      \end{scope}
     \end{tikzpicture}}
\hspace{3mm}
\subfigure[Two periodic tasks with period 2 and relative deadline 1,
with respective phase 1 and 2. Time constraints put $a$ and $b$ in
mutual exclusion: $a$ can execute every \mbox{$[2*k+1,2*k+2[$} and $b$
every \mbox{$[2*k,2*k+1[$}.]{
    \begin{tikzpicture}[yscale=0.9,xscale=1.2]    
      \begin{scope}
        \node[after,label=below:{\underline 1}] (N1) at (0,0) {};
        \node[noconstraint] (N2) at (1,0) {};
        \node[before,label=-150:{\underline 1}] (N3) at (2,0) {};
        \node[synchronisationpoint,label=below:{\underline 2}] (N4) at (1,-1) {};

        \path[block] 
        (N1) edge (N2)
        (N2) edge node[above] {$a$} (N3)
        (N3.south) edge[bend left=70] (N4.east)
        (N4.west) edge[bend left=45] (N2);
      \end{scope}

      \begin{scope}[xscale=1.2,xshift=2.8cm]
        \node[after,label=below:{\underline 2}] (N1) at (0,0) {};
        \node[noconstraint] (N2) at (1,0) {};
        \node[before,label=-150:{\underline 1}] (N3) at (2,0) {};
        \node[synchronisationpoint,label=below:{\underline 2}] (N4) at (1,-1) {};

        \path[block] 
        (N1) edge (N2)
        (N2) edge node[above] {$b$} (N3)
        (N3.south) edge[bend left=70]  (N4.east)
        (N4.west) edge[bend left=45] (N2);
      \end{scope}
     \end{tikzpicture}}}

\centering

\subfigure[Synchronisation using time: $s$ sends a message before time
$3$, $r$ receives it after time $3$. This guarantees that the message
will always be received.]{
    \begin{tikzpicture}    

      \begin{scope}[xshift=-2cm]
        \path (-1.2,0) -- (8.2,0);

        \node[noconstraint] (N0) at (0,0) {};
        \node[after,label=below:3] (N1) at (1,0) {};
        \node[noconstraint] (N2) at (2,0) {};
      \end{scope}

        \begin{scope}
          \node[noconstraint] (N3) at (3,0) {};
          \node[before,label=below:3] (N4) at (4,0) {};
          \node[noconstraint] (N5) at (5,0) {};
        \end{scope}
        \path[block] (N0) edge (N1)
        (N1) edge node[above] {$r$} (N2)
        (N3) edge node[above] {$s$} (N4)
        (N4) edge (N5);

     \end{tikzpicture}}

   \caption{Example uses of time-constrained automata}
   \label{fig:example-time-constrained-automata}
\end{figure}
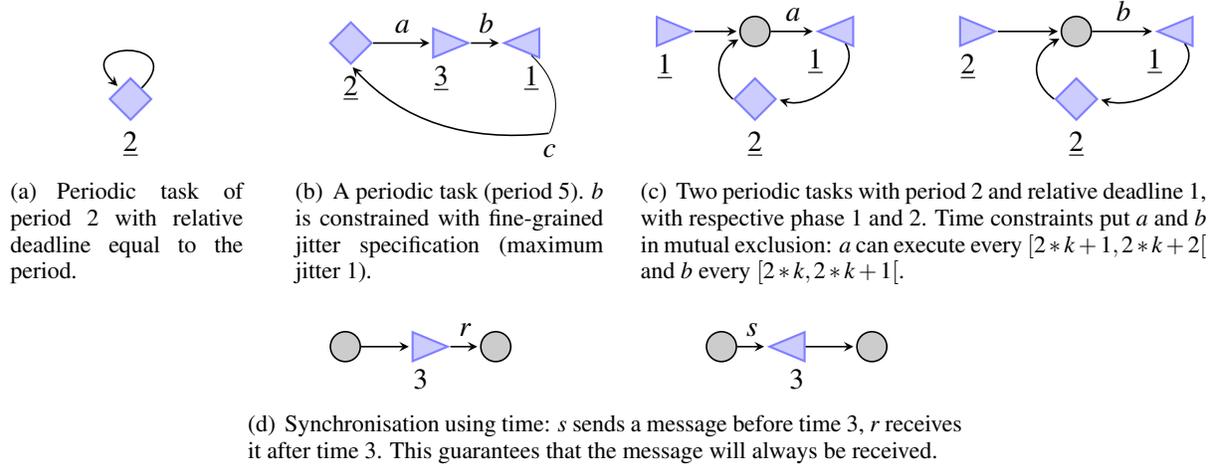

Time-constrained automata can be used to accurately model tasks timing
requirements in a great variety of situations, as it is a superset of
existing time-triggered models.

The basic example is modeling periodic tasks with deadline equal to
their period \cite{liu73scheduling}
(Figure~\ref{fig:example-time-constrained-automata}(a)). More evolved
usage are general periodic tasks, whose deadlines can be smaller or
even larger than periods
(Figure~\ref{fig:example-time-constrained-automata}(c)).
Fine-grained maximum jitter may also be specified
(Figure~\ref{fig:example-time-constrained-automata}(b)).

A periodic task might require an initialization stage, which may have
time constraints (this may be the case, e.g., for device
initialization). This also allows to phase different periodic tasks,
as in Figure~\ref{fig:example-time-constrained-automata}(c), which
allows implementation of mutual exclusion based on time.

The before and after constraints are also well-suited to synchronize
tasks that must communicate.
Figure~\ref{fig:example-time-constrained-automata}(d) gives an example
of a sender and a receiver; the problem is detailed in
section~\ref{sec:safe-inter-time}.

Other applications include going into a degraded mode, or
communicating with device that have complex timing requirements. For
instance we easily developed a mouse and keyboard driver in OASIS that
work without using hardware interrupts. Different stages (detection of
plug/unplug, initialization, normal polling) have different timing
requirements, which can be described accurately with time-constrained
automata.

To sum up, time-constrained automata are a powerful tool to specify
time-constrained tasks that have timing requirements. But they can be
more than that: in the next sections, we show how these automata can
be derived from source code and be used for scheduling.

\subsection{Writing time-constrained programs}
\label{sec:writing-time-constrained-programs}

\psic{} is a programming language designed for implementing
time-constrained automata. \psic{} preserves the operational semantics
of C, but adds time constraints to these semantics with the $\Psi$
extension (this extension could be applied to any imperative
programming language).

C control flow graphs are automata, so C's instructions for control
flow can be used to express sequencing of blocks, loops, and choices.
The basic $\Psi$ addition to C is the addition of \texttt{before},
\texttt{after}, and \texttt{advance} instructions that respectively
add before and after constraints, and synchronization points. It then
becomes possible to express time-constrained automata in $\Psi$C.
There are other extensions to \psic{}, to express for instance
communication between agents (with automatic buffer sizing) and
synchronization between different clocks.

Figure~\ref{fig:code-excerpt-cfg} gives a \psic{} code excerpt with
the corresponding automaton. (Note that our automata differs slightly
from usual control flow graphs because code is carried by arcs,
instead of nodes.)

\begin{figure}[htbp]
  \hbox{\hspace{3cm}
  \begin{minipage}{0.40\linewidth}

\begin{flushleft}
\texttt{while(1) \{}\\
\texttt{\ \ after(1);}\\
\texttt{\ \ if(...) \{}\\
\texttt{\ \ \ \ after(2);}\\
\texttt{\ \ \ \ }\textit{a:}\\
\texttt{\ \ \ \ for(i=0;i<10;i++) }\\
\texttt{\ \ \ \ \ \ \{ }\textit{b:}\\
\texttt{\ \ \ \ \ \ \ \ advance(1);\ }\\
\texttt{\ \ \ \ \ \ \ \ }\textit{c: }\texttt{\}}\\
\texttt{\ \ }\textit{d: }\texttt{\}}\\
\texttt{\ \ else before(2);}\\
\texttt{\ \ }\textit{e:}\\
\texttt{\ \ advance(5);}\\
\texttt{\}}\\
\end{flushleft}

    
  \end{minipage}
  \hfill
  \begin{minipage}{0.6\linewidth}
{
    \begin{tikzpicture}[xscale=1.05,yscale=0.8]
      \node[noconstraint,label=right:while] (while) at (0,0) {};
      \node[after,rotate=-40,label=8:after(1)] (after1) at (1,-1) {};
      \node[noconstraint,label=right:if] (if) at (2,-2) {};
      \node[after,rotate=-40,label=8:after(2)] (after2) at (3,-3) {};
      \node[noconstraint,label=10:for] (for) at (4,-4) {};
      \node[synchronisationpoint,label=below:{advance(1)}] (advance1) at (5,-5) {};
      \node[before,rotate=-45,label=188:before(2)] (before2) at (2,-4) {};
      \node[noconstraint] (joinif) at (3,-5) {};
      \node[synchronisationpoint,label=below:{advance(5)}] (advance5) at (3,-6) {};

      \draw[block] 
      (advance5) edge[bend left=30] (while)
      (while) edge (after1)
      (after1) edge (if)
      (if) edge (after2)
      (if) edge[bend right] (before2)
      (after2) edge node[below left] {\it a} (for)
      (for) edge[bend left=45] node[above right,near end] {\it b} (advance1)
      (advance1) edge[bend left=45] node[below left] {\it c} (for)
      (for) edge node[below] {\it d} (joinif)
      (before2) edge (joinif)
      (joinif) edge node[right] {\it e} (advance5);
    \end{tikzpicture}}

  \end{minipage}}

\caption{A \psic{} code excerpt and corresponding automaton. Blocks
  \textit{a, b, c, d,} and \textit{e} are labelled in both the code
  and the automaton. As an example ``\textit{a;b}'' and
  ``\textit{c;b}'' must be executed within 1 unit of time;
  ``\textit{c;d;e}'' must be executed within 5. }
\label{fig:code-excerpt-cfg}
\end{figure}
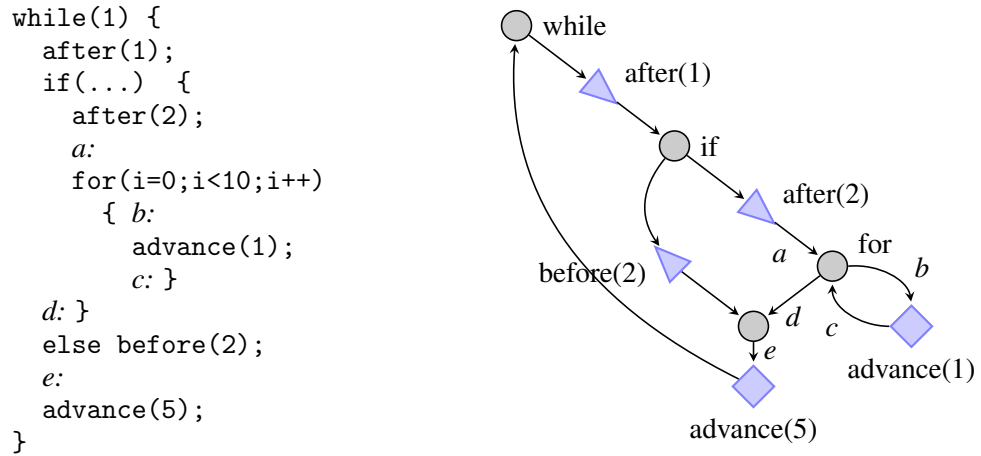

\subsection{Safe interaction of time-constrained tasks}
\label{sec:safe-inter-time}

One of the most interesting aspects of the time-constrained
methodology is that it allows to define communication primitives for
safe interactions between tasks. 

\paragraph{Classical problems of task interactions}

Common problems in communication and synchronization primitives
include:

\begin{itemize*}
\item Deadlock, happening when synchronizing communications (mutex,
  synchronous IPC...) happen in the wrong order.
\item Change in scheduling, happening for every synchronizing
  communication. These makes scheduling and schedulability analysis
  much more complex \cite{lemerre09communication} (e.g. to take into
  account priority inheritance).
\item Nondeterminism, which can happen for every kind of communication.
  This makes executions hard or impossible to reproduce, rendering
  tests useless.
\item Insufficient buffer sizes, which happens for buffered
  communication. This can lead to unpredictable blocking time.
\end{itemize*}

The communication primitives we provide, use time constraints to
avoid all these pitfalls:

\paragraph{Synchronization}

First, we do not provide any synchronization or synchronizing
primitive: all communications are asynchronous, but their ordering is
controlled by the time constraints. For instance,
Figure~\ref{fig:example-time-constrained-automata}(d) shows how we can
enforce block $r$ to happen after block $s$, by choosing an arbitrary
date (3) for ``synchronization''. This slightly over-constrains the
system; but the benefits of doing so (simple optimal scheduling,
simple schedulability analysis) also lead to a gain of performance,
and we think this is a good trade-off.

\paragraph{Determinism}

Second, message communication can be made deterministic. The first
source of nondeterminism occurs because of ordering problems between
the sender and the receiver. See for instance communication using a
variable in shared memory: if the variable is written by the sender
before it is read by the receiver, then communication happens; if the
order is reversed, it does not happen. Using time constraints, the
sending block can be enforced to end before the receiving block
begins, using a ``synchronization date'' as in
Figure~\ref{fig:example-time-constrained-automata}(d). But nothing
prevents another block in the receiver to try to read the value before
this synchronization date, an operation that nondeterministically
succeeds or fails.

This problem is solved by introducing the \emph{visibility date}
concept. A communication can be made only if the receiving block is
always after the visibility date (i.e. the receiving block succeeds
an after node whose date is greater than the visibility date), and if
the sending block is always before the visibility date (i.e. it
precedes a before node whose date is smaller than the visibility
date). This is achieved by tying the communication primitives to the
time constraints, and using an appropriate implementation of the
communication primitives. 

The second source of nondeterminism occurs when multiple agents send a
message to another agent at nearly the same time. If the reception
depends on sending order, message reception can be nondeterministic.
This can be solved using appropriate implementation, and is largely
independent from the time-constrained model. The combination of these
two techniques allows implementation of communication primitives that
are provably deterministic.

\paragraph{Buffer sizes}

The last possible caveat is buffer size. If a buffer is not large
enough to contain all the messages, a run-time situation can occur
where the sender cannot store the message it wants to send. Then it
must either block (but we do not want to provide synchronizing
primitives) or throw a runtime error, which is difficult to analyze.

Automatic buffer computation efficiently solves this problem. Time
constraints allows to infer the respective rates and phases of
communication production/consumption, which allows to know when buffer
parts can be re-used, and to infer the exact sizing of the buffer.
Details vary according to the communication primitive used.

\paragraph{OASIS primitives for safe interaction}

Thus time constraints are of great help for designing communication
primitives for safe interaction. We have implemented several such
primitives: \emph{temporal variable} implements a 1-to-$n$ regular
data flow, while \emph{message} are $n$-to-1 irregular communication.
Several others are being implemented in the context of the PharOS
project. More details can be found in
\cite{chabrol05deterministic,aussagues09os}.

In practice, the methodology for designing OASIS applications consists
in writing Gantt charts with the timing constraints of the receiver
and sender tasks, and of the communication primitives. The periods and
phases are tuned to respect the end-to-end requirements of the
tasks. But we are currently working on a more formal methodology for
designing OASIS applications.

\subsection{Experience with the time-constrained methodology}

Writing real-time applications using the time-constrained methodology
(and $\Psi$C programming language) allows to write safe, parallel
programs with ease (see \cite{david04oasis} for a large practical
example from the nuclear industry). The specification, design and
implementation are tightly coupled, which greatly simplifies
verification and validation. Verification and validation is generally
the most costly phase when designing a real-time system, especially
when it is safety-critical.

Parallel programs in $\Psi$C are \emph{deterministic}, i.e. have
predictable and reproducible execution. Hence, tests are reproducible
despite the parallel execution.

Time-constrained tasks synchronize only using time, and do not use
mutexes or semaphores. This allows us to perform \emph{exact}
feasibility analysis, and to reach high processor utilization, even on
multiprocessors. This also provides safety guarantees (deadlock is
impossible).

\vspace{-4mm}
\section{Scheduling of time-constrained tasks}
\label{sec:scheduling-time-constrained-tasks}

This section presents the precise scheduling semantics of the
time-constrained automata, and gives an optimal scheduling algorithm
on single processor based on EDF. We decompose again the presentation
in three parts: scheduling of chains, trees, and automata.

\subsection{Definitions}

\paragraph{Required execution time function}

We assume the existence of a function $||\cdot|| : Blocks \to
\mathbb{R}$ that gives the execution time necessary to complete a
block. Note that this function is not necessary to define the
semantics of TCA.

\paragraph{Validity and correctness}

We say that a schedule is \emph{valid} when it respects the semantics
of the task model (e.g. executes a job only between its start time and
deadline). We say that it is \emph{correct} when it is valid and tasks
have enough time to complete before deadline.

\paragraph{Feasibility and optimal algorithm} 

A set of tasks is \emph{feasible} if there exists a correct
schedule. An \emph{optimal} scheduling algorithm finds a correct
schedule whenever one exists (i.e. whenever the set of tasks is
feasible).

\subsection{Scheduling of time-constrained chains}
\subsubsection{Semantics}

\paragraph{Schedule mapping}

On single processor computers, a schedule is a function $s :
Blocks \times \mathbb{R}_{+} \to \{0,1\}$ that tells whether block
$b$ is scheduled at time $t$.

In an interval of time $[t_{1},t_{2}]$, a block $b$ is executed for a
duration of $\int_{t_{1}}^{t_{2}} s(b,t) \mathrm{d}t$

For multiprocessor systems, the definition is the same, because the
specific placement on the processors can be
abstracted~\cite{lemerre08equivalence}.

\paragraph{Conditions for a valid schedule of time-constrained chains}

Conditions for a function $s$ to be a valid schedule express the
sequentiality of blocks and the time constraints in the schedule:

\begin{itemize*}
\item If a block $b$ has an ``after'' constraint of date $d$, it must
  not be scheduled before $d$: 
  \[ \forall t,\quad t < d \ \Rightarrow\ s(b,t) = 0 \]
\item If a block $b$ has a ``before'' constraint of date $d'$, it must
  not be scheduled after $d$:
  \[ \forall t,\quad t > d' \ \Rightarrow\  s(b,t) = 0\]
\item If block $b$ precedes block $b'$, it must be scheduled before
  $b'$: (note: the two formulas are equivalent)

  \[ \forall t',\quad \big( s(b',t') \ne 0 \ \Rightarrow\ \forall t > t',\ s(b,t) = 0  \big) \quad\iff\quad \forall t,\quad \big( s(b,t) \ne 0 \ \Rightarrow\ \forall t' < t,\ s(b',t') = 0  \big) \]

\end{itemize*}

\paragraph{Condition for a correct schedule}

A correct schedule is a valid schedule, with the condition that all
blocks $b$ must be scheduled for at least their required execution
time:
  \[ \int_{t=0}^{+\infty} \!\!\!s(b,t) \mathrm{d}t \ \ge \ ||b|| \]

  Note the use of $\ge$ instead of $=$ in the previous equation: this
  allows feasibility analysis to allocate more CPU time than
  necessary.

\subsubsection{Optimal scheduling algorithm on single processors}

Current real-time scheduling algorithms consider tasks as releasing a
sequence of static jobs \cite{baruah04schedulingreal-time}. This is
not applicable to our model because a time-constrained task has
``changing jobs'', i.e. job's execution time and deadline can change
dynamically.

We propose EDF-dyn, an extension to EDF to schedule a set of
time-constrained chains:

\begin{definition}[EDF-dyn]
  EDF-dyn schedules at each instant $t$, the task whose current
  block's implicit deadline is the soonest, chosen among all tasks
  whose current block's implicit start date is sooner than the current
  date. Ties can be broken arbitrarily.
\end{definition}

For a set of chains $\mathcal{C}$, we note
${\texttt{EDF-dyn}}(\mathcal{C})$ the schedule produced by EDF-dyn.

The proof that EDF-dyn is optimal on time-constrained chains can be
readily adapted from the proof that EDF is optimal on static jobs:

\begin{theorem}
  EDF-dyn is optimal for time-constrained chains on single processors :
  if $\mathcal{C}$ is feasible, then $\mathtt{EDF-dyn}(\mathcal{C})$
  is a correct schedule.
\end{theorem}

\begin{proof}[Proof sketch] 
  The proof is the same as the standard EDF optimality proof
  (e.g.~\cite{baruah04schedulingreal-time}), obtained by replacing
  ``job'' by ``block'' and ``deadline'' by ``implicit deadline''. The
  proof is by absurd: if EDF-dyn is not optimal, and $t_0$ is the last
  instant where no more correct schedule agrees with EDF-dyn, we
  construct a correct schedule that agrees with EDF-dyn until $t_o +
  \delta$.
\end{proof}

EDF-dyn can be efficiently implemented, as seen in the next section.

\subsubsection{Implementation}

\paragraph{High-level algorithm}
\label{sec:formal-algorithm}

\begin{itemize*}
\item When a scheduling decision is made, the algorithm selects
  the task on the ready list with the smallest current deadline;
\item The scheduler is awaken only when:
  \begin{enumerate*}
  \item the current date becomes the same as the date of some
    ``after'' node; the current task is added to the ready list, and a
    new decision is made;
  \item an ``after'' node is reached, whose date is bigger than the
    current date; the current task is removed from the ready list, and
    a decision is made;
  \item a ``before node'' is reached. The ``current deadline'' of the
    task is changed, and a decision is made.
  \end{enumerate*}
\end{itemize*}

Intuitively, case 1. represent new job becoming available; case 2. the
fact that a task has to wait before continuing execution; case 3. that
execution of a blocks finishes before its deadline.

\begin{note} This algorithm behaves as standard EDF on regular jobs. \end{note}

\begin{theorem}
  This algorithm implements EDF-dyn.
\end{theorem}

\begin{proof}[Proof sketch] The instants where the scheduler is awaken
  are sufficient to make the ready list and ``current deadline'' of
  the tasks always up-to-date.
\end{proof}

\paragraph{Implementation in OASIS}

EDF-dyn is implemented in the OASIS
kernel\cite{chabrol05deterministic}. To wake the scheduler when
``after'' and ``before'' nodes are reached (corresponding to
\texttt{after} and \texttt{before} $\Psi$ instructions in source
code), these are replaced by system calls at compilation.

The chains are simplified using
Theorem~\ref{th:useless-undoable-constraints}, and ``before'' nodes
carry information about the deadline of the following ``before'' node,
so that the next deadline is known without complex lookup. Thus,
``before'' nodes are transformed into ``update deadline'' system
calls.

Finally, we also check that deadlines are not missed by waking the
scheduler up at deadline dates; different blocks can be taken upon
deadline miss (from system stop to degraded mode). We can also check
that a block $b$ is not executed more than $||b||$.

\subsection{Scheduling of time-constrained trees}

\subsubsection{Semantics}

\paragraph{From trees to chains}

A path in a time-constrained tree is a time-constrained chain. Given a set
of trees $\mathcal{T}$ and a set of choices in the tree $\mathcal{U}$,
we define the operator $\mathcal{E}$ so that
$\mathcal{E}(\mathcal{U},\mathcal{T}) = \mathcal{C}$ is the set of
time-constrained chains obtained by extracting from the trees only the
paths corresponding to the choices in $\mathcal{U}$.

\paragraph{Tree schedule}

The schedule for a tree has to take into account the fact that
depending on the choices made, multiple time constraints are
possible; so a schedule for a tree is a function of the choices made.

We refer to the instant when the block preceding the choice node has
finished executing to the \emph{instant of choice}. As the choice is
not known by the scheduling algorithm until this instant, the
schedules must be the same up to it.

More formally, we define a valid schedule for a set of trees
$\mathcal{T}$ to be a function $S$ as such:

\begin{itemize*}
\item $S$ associates to each possible set of choices $\mathcal{U}$,
  $S_{\mathcal{U}}$ which is a valid schedule mapping for the set of
  chains $\mathcal{E}(\mathcal{U},\mathcal{T})$;
\item Additionally, for two different sets of choices $\mathcal{U}$
  and $\mathcal{U'}$, $S_{\mathcal{U}}$ and $S_{\mathcal{U}}$ must
  coincide up to the time $\tau$ when the first choice node that
  differentiates them is reached:
\[ \forall t < \tau,\ \forall b, \quad S_{\mathcal{U}}(b,t) = S_{\mathcal{U'}}(b,t) \]
\end{itemize*}

Figure~\ref{fig:correct-tree-schedule} gives an example.

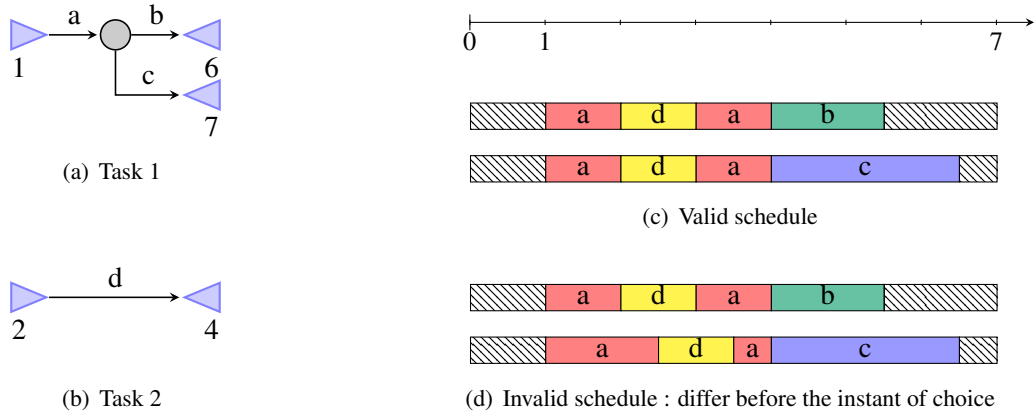
\begin{figure}[htbp]

    
\hbox{\hspace{1.5cm}\vbox{\subfigure[Task 1]{\label{fig:task-1-correct-tree-schedule}
  \begin{tikzpicture}[xscale=1.6, scale=0.8]

    \node[after,label=below:1] (1) at (0,0) {};
    \node[noconstraint] (2) at (1,0) {};
    \node[before,label=below:6] (5) at (2,0) {};
    \node[before,label=below:7] (7) at (2,-1) {};

    \draw[block] (1) -- node[above] {a} (2);
    \draw[block] (2) -- node[above] {b} (5);
    \draw[block] (2) |- node[above,near end] {c} (7);
  
  \end{tikzpicture}}\vspace{6mm}

  \subfigure[Task 2]{
  \begin{tikzpicture}[xscale=1.6, scale=0.8]
    \node[after,label=below:2] (1) at (0,0) {};
    \node[before,label=below:4] (5) at (2,0) {};

    \node[before,white] (7) at (2,-1) {};

    \draw[block] (1) -- node[above] {d} (5);
  
  \end{tikzpicture}}}
\hspace{-10cm}\vbox{
\begin{tikzpicture}[xscale=1]
  \draw[white] (0,0) -- (7.5,0);
  \echellechronogramme{7};
\end{tikzpicture}

\subfigure[Valid schedule]{\label{fig:valid-tree-schedule}
      \begin{tikzpicture}[xscale=1, yscale=0.7]
  \draw[white] (0,0) -- (6.5,0);
    \begin{schedule}
      \newschedule; \emptyplage{1}; \newplage{a}{1}; \newplage{d}{1}; \newplage{a}{1}; \newplage{b}{1.5}; \emptyplage{1.5};
      \newschedule; \emptyplage{1}; \newplage{a}{1}; \newplage{d}{1}; \newplage{a}{1}; \newplage{c}{2.5}; \emptyplage{0.5};
    \end{schedule}
    \end{tikzpicture}}

\subfigure[Invalid schedule : differ before the instant of choice]{
      \begin{tikzpicture}[xscale=1, yscale=0.7]
  \draw[white] (0,0) -- (6.5,0);
     \begin{schedule}
      \newschedule; \emptyplage{1}; \newplage{a}{1}; \newplage{d}{1}; \newplage{a}{1}; \newplage{b}{1.5};  \emptyplage{1.5};
      \newschedule; \emptyplage{1}; \newplage{a}{1.5}; \newplage{d}{1}; \newplage{a}{0.5}; \newplage{c}{2.5}; \emptyplage{0.5};
    \end{schedule}
    \end{tikzpicture}}}}

  \caption{Here, $||a|| = 2$, $||b||=2$, $||c||=1$ and $||d||=1$;
    there is one choice so a tree schedule is two schedule mapping.
    Instant of choice (i.e. last time $a$ is executed) is at time
    $4$.}
\label{fig:correct-tree-schedule}
\end{figure}

\paragraph{Correctness conditions}

We define a schedule $S$ to be correct for a set of trees
$\mathcal{T}$ when all schedules on the corresponding chains are
correct: \[ \forall \mathcal{U},\quad
S_{\mathcal{U}} \mathrm{\ is\ correct.}\]

In Figure~\ref{fig:correct-tree-schedule}, the
schedule~\ref{fig:valid-tree-schedule} is correct because if block $c$
is chosen, the first schedule mapping is correct for the chains $a \to
c$ and $d$; and if block $b$ is chosen, the second is correct for the
chains $a \to b$ and $d$.

This means that all possible executions can be correctly scheduled.

\subsubsection{Choice deadline inheritance}

Online, deadline-based algorithms such as EDF always need to know the
next deadline; when scheduling a tree, this ``next deadline'' is not
known because it depends on future execution. Choice deadline
inheritance solves this problem.

We define \emph{choice deadline inheritance} as follows : for each
``choice'' node, we consider the (temporally) closest ``before'' node
in all possible following nodes. If $\tau$ is the date of this
``before'' node, we add to the choice node a ``before'' constraint of
date $\tau$.

\begin{note}
  If there are no following ``before'' nodes, we do nothing. This can
  be viewed as inheriting a deadline at date $+\infty$.
\end{note}

On a set of trees $\mathcal{T}$, we denote by
$\texttt{CDI}(\mathcal{T})$ the set of trees transformed with choice
deadline inheritance.  Figure~\ref{fig:transformation-with-cdi}
provides an example.

\begin{figure}[htbp]
  \centering

  \begin{tikzpicture}[xscale=1.8]

    \node[after,label=below:1] (1) at (0,0) {};
    \node[before,label=-120:6] (2) at (1,0) {};
    \node[before,label=below:6] (5) at (2,0) {};
    \node[before,label=below:7] (7) at (2,-1) {};

    \draw[block] (1) -- node[above] {a} (2);
    \draw[block] (2) -- node[above] {b} (5);
    \draw[block] (2) |- node[above,near end] {c} (7);

  
  \end{tikzpicture}

  \caption{Figure~\ref{fig:task-1-correct-tree-schedule} transformed
    with \texttt{CDI}. The smallest deadline is 6, so the choice node
    ``inherits'' this deadline.}
  \label{fig:transformation-with-cdi}
\end{figure}
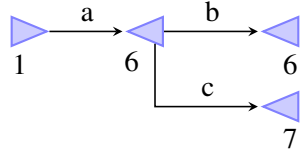

The following theorem shows that choice deadline inheritance, which
adds constraints to the time-constrained trees, does not change its
schedulability:

\begin{theorem}
  A schedule is correct for a set of time-constrained trees
  $\mathcal{T}$ iff it is correct for $\texttt{CDI}(\mathcal{T})$.
\end{theorem}

\begin{proof}
  As $\texttt{CDI}$ add constraints, any schedule correct for
  $\texttt{CDI}(\mathcal{T})$ is correct for $\mathcal{T}$.

  If a schedule is correct for $\mathcal{T}$, we consider a path of
  the tree choosing the soonest deadline. For this path, the block
  before the choice node is implicitly constrained to finish before
  this deadline. And because of the validity condition on trees, all
  schedules for the tree must be the same until this choice node is
  reached, so they all finish before this deadline. Thus the
  additional constraint expressed by $\texttt{CDI}(\mathcal{T})$ is
  fulfilled.
\end{proof}

This proves that the constraints added by choice deadline inheritance
do not affect feasibility. We will now see how they are used to
implement EDF scheduling on trees.

\subsubsection{Optimal scheduling on single processor}

We now define the EDF-dyn-min algorithm (EDF-dyn with minimal
deadline) algorithm on time-con\-strained trees:

\begin{definition}
  EDF-dyn-min schedules at each instant $t$, the task whose current
  block's \emph{possible} implicit deadline is the soonest, chosen
  among all tasks whose current block's implicit start date is sooner
  than the current date. Ties can be broken arbitrarily.
\end{definition}

The schedule of Figure~\ref{fig:valid-tree-schedule} is the schedule
obtained with EDF-dyn-min.

Lemma~\ref{th:eq-edf-dyn-min-edf-dyn-cdi} is central to the online
implementation of EDF-dyn-min:

\begin{lemma}
  \label{th:eq-edf-dyn-min-edf-dyn-cdi}
  Let $\mathcal{T}$ be a set of trees, and $\mathcal{U}$ a set of
  choices for these trees. Then

  \[ {\texttt{EDF-dyn-min}}(\mathcal{T})_{\mathcal{U}} \ =\ 
  {\texttt{EDF-dyn}}(\,\mathcal{E}(\mathcal{U},
  \texttt{CDI}(\mathcal{T}))) \]

  That is, for a given set of choices $\mathcal{U}$, the schedule of
  EDF-dyn-min can be obtained by executing $\texttt{EDF-dyn}$ on the
  paths of $\texttt{CDI}(\mathcal{T})$ corresponding to $\mathcal{U}$.
\end{lemma}

\begin{proof}[Proof sketch] The soonest possible implicit deadline in
  the definition of EDF-dyn-min is the constraint's date added by the
  choice deadline inheritance algorithm. Thus for each possible path,
  EDF-dyn is followed with this additional constraint on choices.
\end{proof}

\begin{theorem}
  EDF-dyn-min is optimal on time-constrained trees.
\end{theorem}
\begin{proof}[Proof sketch]
  Let $\mathcal{T}$ be a feasible set of time-constrained trees and $S =
  \texttt{EDF-dyn-min}(\mathcal{T})$ be the schedule produced by
  EDF-dyn-min.

  For any set of choices $\mathcal{U}$, by
  Lemma~\ref{th:eq-edf-dyn-min-edf-dyn-cdi} and because EDF-dyn is
  optimal, $S_{\mathcal{U}}$ is a correct schedule for
  $\mathcal{E}(\mathcal{U},\mathcal{T})$.

  It remains to prove that for two choices $\mathcal{U}$ and
  $\mathcal{U'}$, the schedules remain the same until the first
  instant of choice. Following
  Lemma~\ref{th:eq-edf-dyn-min-edf-dyn-cdi} and definition of choice
  deadline inheritance, all deadlines are the same until the first
  differing choice is taken; thus the schedules are the same until
  this moment.
\end{proof}

\paragraph{Implementation in OASIS}

The implementation of EDF-dyn-min in OASIS executes only trees with
choice deadline inheritance. By
Lemma~\ref{th:eq-edf-dyn-min-edf-dyn-cdi}, it uses the EDF-dyn
algorithm given in Section~\ref{sec:formal-algorithm}.

To implement choice and choice deadline inheritance, the compiler adds
two separate ``update'' system calls, at the beginning of each
``then'' and ``else'' branches of each choice that changes timing
behavior.

In fact, the ``update'' system call for the branch with the soonest
deadline is redundant and can be removed.

\subsection{Scheduling of time-constrained automata}

The semantics of scheduling a set of time-constrained automata
$\mathcal{A}$ is the semantics of scheduling the corresponding set of
unfolded trees.

\paragraph{Implementation in OASIS}

We use the same algorithm as for time-constrained trees: we apply choice
deadline inheritance to the automaton and use the EDF-dyn algorithm.

By replacing \texttt{after} instructions by after system calls, and
\texttt{before} instructions by ``update deadline'' system calls, the
time-constrained automaton is completely embedded in the code structure.
Execution of this code unfolds the automaton on the fly. Conversion
from relative to absolute labelling is also performed dynamically.

\vspace{-3mm}

\paragraph{Feasibility analysis}

Although we won't present it due to lack of space, it is possible to
conduct a feasibility analysis on time-constrained tasks. This is done
by computing the product of the automata to analyze the blocks
simultaneously executed, and translate this product into a linear
programming (in fact, network flow) problem (variables represent the
amount of a block done on an interval). Although of great complexity
in the worst case, this proved to be very tractable for all the
industrial problems we have considered.

One of the reason why the problem is tractable is the absence of
interaction between communication and scheduling. These interactions
are often hard to take into account in practice.

\section{Conclusion}

The time-constrained task model is a general model able to accurately
describe the temporal behavior of algorithms constrained by time. One
of its main interest is that it can be used for scheduling hard
real-time tasks, with high efficiency since optimal scheduling exists.
We have presented their particular scheduling semantics, and some of
their uses (guaranteed safe interaction, deterministic communication,
and derivation from source code).

We are currently writing a complete formalization of the scheduling
semantics, as well as very detailed proof of the optimality of our
scheduling and feasibility analysis algorithms. This will allow to
provide formal, abstract semantics for the time-constrained automata
and the different communication primitives, for use by future
model-checking tools.

A forthcoming paper will present feasibility analysis in detail, as
well as results on multiprocessor scheduling of our model. Indeed, our
non-blocking model of tasks can be executed with very high
utilization, even on multiprocessor computers, which makes it very
promising for high-performance real-time systems. Even if ``the
future'' needs to be known for optimal scheduling
\cite{dertouzos89multiprocessor}, the time-constrained model only
expresses a reasonable set of possible futures, thus allowing
near-optimal scheduling on multiple processors.

\paragraph{Present and future extensions}

The model presented here is sufficient to study schedulability, but
some extensions make it more practical to build real applications
using a time-constrained model. Among these are communication
primitives, synchronization on multiple clocks (that allow a task to
synchronize with another one that has a different rate), variable
afters (the date is only known to be in some interval)\ldots All these
extensions are implemented in the industrial version of OASIS.

A future extension of particular interest with regard to scheduling is
the extension of the model to allow multithreaded computations. This
is done by allowing certain nodes to ``create'' new automata (a
\texttt{fork()}-like node), and some other to wait for their
completion (a \texttt{wait()}-like node). Another future extension is
support for event-triggered computation in the model.

A last research direction is the study of automatic transformation
from formal specifications (e.g. timed automata \cite{alur94theory})
to our task model.

{\small
\bibliographystyle{ieeetr}
\bibliography{/home/ml210990/doc/papers/bibliographie.bib}

}

\end{document}